%% file: IZS2011.tex
\documentclass[a4paper,conference]{IEEEtran}
\usepackage{graphicx,ifthen}
\usepackage{psfrag,amssymb,amsthm}
\usepackage[cmex10]{amsmath}
\usepackage{cite,flushend,color,pst-plot,stfloats,pst-sigsys}
\usepackage{pstricks-add}

\newtheorem{thm}{Theorem}

\newtheorem{cor}{Corollary}

\theoremstyle{definition}
\newtheorem{definition}{Definition}

\def \arxiv {1}

\title{On the Information Loss in Memoryless Systems: The Multivariate Case}

\author{\IEEEauthorblockN{Bernhard C. Geiger\IEEEauthorrefmark{1}, Gernot Kubin\IEEEauthorrefmark{1}
\IEEEauthorblockA{\IEEEauthorrefmark{1}Signal Processing and Speech Communication Laboratory, Graz University of Technology, Austria}
$\{$geiger,gernot.kubin$\}$@tugraz.at}}

\begin{document}
\newcounter{myTempCnt}

\ifthenelse{\arxiv=1}{\input{abbrevations_processing.tex}}{\input{../../Blocks/abbrevations_processing.tex}}

\maketitle

\begin{abstract}
In this work we give a concise definition of information loss from a system-theoretic point of view. Based on this definition, we analyze the information loss in memoryless input-output systems subject to a continuous-valued input. For a certain class of multiple-input, multiple-output systems the information loss is quantified. An interpretation of this loss is accompanied by upper bounds which are simple to evaluate.

Finally, a class of systems is identified for which the information loss is necessarily infinite. Quantizers and limiters are shown to belong to this class.
\end{abstract}

\section{Introduction}\label{sec:intro}
In the XXXI. Shannon lecture Han argued that information theory links information-theoretic quantities, such as entropy and mutual information, to operational quantities such as source, channel, capacity, and error probability~\cite{Han_Musing}. In this work we try to make a new link to an operational quantity not mentioned by Han: information loss. Information can be lost, on the one hand, in erasures or due to superposition of noise as it is known from communication theory. Dating back to Shannon~\cite{Shannon_TheoryOfComm} this loss is linked to the conditional entropy of the input given the output, at least in discrete-amplitude, memoryless settings. On the other hand, as stated by the data processing inequality (DPI,~\cite{Gray_Entropy}), information can be lost in deterministic, noiseless systems. It is this kind of loss that we will treat in this work, and we will show that it makes sense to link it to the same information-theoretic quantity.
% We will show that it makes sense to link this kind of loss to the same information-theoretical quantity.

The information loss in input-output systems is very sparsely covered in the literature. Aside from the DPI for discrete random variables (RV) and static systems, some results are available for jointly stationary stochastic processes~\cite{Pinsker_InfoEngl}. Yet, all these results just state that \emph{information is lost}, without quantifying this loss. Only in~\cite{Watanabe_InfoLoss} the information lost by collapsing states of a discrete-valued stochastic process is quantified as the difference between the entropy rates at the input and the output of the memoryless system.

Conversely, energy loss in input-output systems has been deeply analyzed, leading to meaningful definitions of transfer functions and notions of passivity, stability, and losslessness. Essentially, it is our aim to develop a system theory not from an energetic, but from an information-theoretic point of view. So far we analyzed the information loss of discrete-valued stationary stochastic processes in finite-dimensional dynamical input-output systems~\cite{Geiger_NLDyn1starXiv}, where we proposed an upper bound on the information loss and identified a class of information-preserving systems (the information-theoretic counterpart to lossless systems). In~\cite{Geiger_ISIT2011arXiv} the information loss of continuous RVs in memoryless systems was quantified and bounded in a preliminary way. In this work, extending~\cite{Geiger_ISIT2011arXiv}, we analyze the information loss for static multiple-input, multiple-output systems which are subject to a continuous input RV. Unlike in our previous work, we permit functions which lose an infinite amount of information and present the according conditions. Aside from that we provide a link between information loss and differential entropy, a quantity which is not invariant under changes of variables. The next steps towards an information-centered system theory are the analysis of discrete-time dynamical systems with continuous-valued stationary input processes and a treatment of information loss in multirate systems.

In the remainder of this paper we give a mathematically concise definition of information loss (Section~\ref{sec:defLoss}). After restricting the class of systems in Section~\ref{sec:problem}, in Section~\ref{sec:mainResults} we provide exact results for information loss together with simple bounds, and establish a link to differential entropies. Finally, in Section~\ref{sec:extension} we show under which conditions the information loss becomes infinite. \ifthenelse{\arxiv=1}{

This manuscript is an extended version of a paper submitted to a conference.}{An extended version of this paper, accompanied by several examples illustrating the theoretical results, is available in~\cite{Geiger_ILStatic_arXiv}.}

\section{A Definition of Information Loss}\label{sec:defLoss}
When talking about the information loss induced by processing of signals, it is of prime importance to accompany this discussion by a well-based definition of information loss going beyond, but without lacking, intuition. \ifthenelse{\arxiv=1}{Further, the definition shall also allow generalizations to stochastic processes and dynamical systems without contradicting previous statements. We try to meet both objectives with the following}{We try to meet this objective with the following}

\ifthenelse{\arxiv=1}{ % arXiv-Definition
\begin{definition}\label{def:loss}
 Let $X$ be an RV\footnote{Note that $X$ and all other involved RVs need not be scalar-valued.} on the samples space $\dom{X}$, and let $Y$ be obtained by transforming $X$. We define the information loss induced by this transform as
\begin{equation}
 L(X\to Y) = \sup_{\partit{}} \left(\mutinf{\hat{X};X}-\mutinf{\hat{X};Y}\right)
\end{equation}
where the supremum is over all partitions $\partit{}$ of $\dom{X}$, and where $\hat{X}$ is obtained by quantizing $X$ according to the partition $\partit{}$ (see Fig.~\ref{fig:sysmod}).
\end{definition}
}
% IZS-Definition
{\begin{definition}\label{def:loss}
 Let $X$ be an RV\footnote{Note that $X$ and all other involved RVs need not be scalar-valued.} on the samples space $\dom{X}$, and let $Y$ be obtained by transforming $X$. We define the information loss induced by this transform as
\begin{equation}
 L(X\to Y) = \sup_{\partit{}} \left(\mutinf{\hat{X};X}-\mutinf{\hat{X};Y}\right)
\end{equation}
where the supremum is over all partitions $\partit{}$ of $\dom{X}$, and where $\hat{X}$ is obtained by quantizing $X$ according to the partition $\partit{}$ (see Fig.~\ref{fig:sysmod}).
\end{definition}
}

This Definition is motivated by the data processing inequality (cf.~\cite{Gray_Entropy}), which states that the expression under the supremum is always non-negative: Information loss is the worst-case reduction of information about $\hat{X}$ induced by transforming $X$. We now try to shed a little more light on Definition~\ref{def:loss} in the following

\begin{figure}[t]
 \centering  
% arXiv picture
\begin{pspicture}[showgrid=false](1,1)(8,3.5)
 	\pssignal(1,2){x}{$\hat{X}$}
 	\pssignal(3,1){n}{$\partit{}$}
	\psfblock[framesize=1 0.75](3,2){oplus}{$Q(\cdot)$}
	\psfblock[framesize=1.5 1](6,2){c}{$g(\cdot)$}
	\pssignal(8,2){y}{$Y$}
  \ncline[style=Arrow]{n}{oplus}
	\ncline[style=Arrow]{oplus}{x}
 \nclist[style=Arrow]{ncline}[naput]{oplus,c $X$,y}
  \ncline[style=Arrow]{c}{oplus}
	\psline[style=Dash](2,2.75)(4,2.75)
	\psline[style=Dash](2,2.25)(2,2.75)
	\psline[style=Dash](4,2.25)(4,2.75)
	\psline[style=Dash](1.75,3.25)(7.25,3.25)
	\psline[style=Dash](1.75,2.25)(1.75,3.25)
	\psline[style=Dash](7.25,2.25)(7.25,3.25)
 	\rput*(3,2.75){\scriptsize{$\mutinf{\hat{X};X}$}}
	\rput*(4.5,3.25){\scriptsize{$\mutinf{\hat{X};Y}$}}
\end{pspicture}
\caption{Model for computing the information loss of a memoryless input-output system $g$. $Q$ is a quantizer with partition $\partit{}$.}
\label{fig:sysmod}
\end{figure}
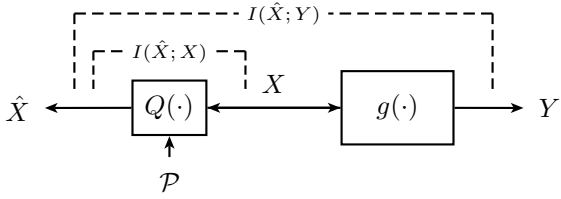

%% Proof of arXiv-Version
\ifthenelse{\arxiv=1}{
\begin{thm}\label{thm:defEq}
The information loss of Definition~\ref{def:loss} is given by:
\begin{IEEEeqnarray}{RCL}
  L(X\to Y) %&=& \sup_{\partit{}} \left(\mutinf{\hat{X};X}-\mutinf{\hat{X};Y}\right)\label{eq:supPart}\\
&=& \lim_{\hat{X}\to X} \left(\mutinf{\hat{X};X}-\mutinf{\hat{X};Y}\right)\label{eq:limPart}\\
&=& \ent{X|Y}\label{eq:diffEnt}
\end{IEEEeqnarray}
\end{thm}

\begin{proof}
 We start by noticing that
\begin{equation}
 \mutinf{\hat{X};X}-\mutinf{\hat{X};Y} = \ent{\hat{X}|Y}\leq\ent{X|Y}\label{eq:boundIL}
\end{equation}
by the definition of mutual information and since both $\hat{X}$ and $Y$ are functions of $X$. The inequality in~\eqref{eq:boundIL} is due to data processing ($\hat{X}$ is a function of $X$). We now show that in the supremum over all partitions equality can be achieved.

To this end, observe that among all partitions of the sample space of $X$ there is a sequence $\{\partit{n}\}$ of increasingly fine partitions\footnote{i.e., $\partit{n+1}$ is a refinement of $\partit{n}$.} such that
\begin{equation}
 \limn\hat{X}_n=X\label{eq:partLim}
\end{equation}
where $\hat{X}_n$ is the quantization of $X$ induced by partition $\partit{n}$.
By the axioms of entropy (e.g.,~\cite[Ch.~14]{Papoulis_Probability}), $\ent{\hat{X}_n|Y}$ is an increasing sequence in $n$ with limit $\ent{X|Y}$. Thus, this limit represents the supremum in Definition~\ref{def:loss}, which proves~\eqref{eq:diffEnt}.

Note further that each converging sequence $\hat{X}\to X$ contains a converging subsequence $\hat{X}_n\to X$ satisfying~\eqref{eq:partLim}, and where $\hat{X}_{n+1}$ is obtained by refining the partition inducing $\hat{X}_n$. Therefore,
\begin{equation}
  \lim_{\hat{X}\to X} \ent{\hat{X}|Y}=\limn \ent{\hat{X}_n|Y} =\ent{X|Y}
\end{equation}
which completes the proof.
\end{proof}
}
%% Proof of IZS-Version
{
\begin{thm}\label{thm:defEq}
The information loss of Definition~\ref{def:loss} is given by the conditional entropy of the input given the output, i.e.,
\begin{equation}
  L(X\to Y) = \ent{X|Y}\label{eq:diffEnt}.
\end{equation}
\end{thm}

\begin{proof}
 We start by noticing that
\begin{equation}
 \mutinf{\hat{X};X}-\mutinf{\hat{X};Y} = \ent{\hat{X}|Y}\leq\ent{X|Y}\label{eq:boundIL}
\end{equation}
by the definition of mutual information and since both $\hat{X}$ and $Y$ are functions of $X$. The inequality in~\eqref{eq:boundIL} is due to data processing ($\hat{X}$ is a function of $X$). We now show that in the supremum over all partitions equality can be achieved.

To this end, observe that among all partitions of the sample space of $X$ there is a sequence $\{\partit{n}\}$ of increasingly fine partitions\footnote{i.e., $\partit{n+1}$ is a refinement of $\partit{n}$.} such that
\begin{equation}
 \limn\hat{X}_n=X\label{eq:partLim}
\end{equation}
where $\hat{X}_n$ is the quantization of $X$ induced by partition $\partit{n}$.
By the axioms of entropy (e.g.,~\cite[Ch.~14]{Papoulis_Probability}), $\ent{\hat{X}_n|Y}$ is an increasing sequence in $n$ with limit $\ent{X|Y}$. Thus, equality in~\eqref{eq:boundIL} is achieved, which completes the proof.
\end{proof}}

This Theorem shows that the supremum in Definition~\ref{def:loss} is achieved for $\hat{X}\equiv X$, i.e., when we compute the difference between the \emph{self-information} of the input and the information the output of the system contains about its input. \ifthenelse{\arxiv=1}{This difference was shown to be identical to the conditional entropy of the input given the output -- the quantity which is also used for quantifying the information loss due to noise or erasures (in the discrete-valued, memoryless case).}{} In addition to that, the Theorem suggests a natural way to measure the information loss via measuring mutual informations, as it is depicted in Fig.~\ref{fig:sysmod}. As we will see later (cf.~Theorem~\ref{thm:WisX}), the considered partition does not have to be infinitely fine, but indeed a comparably coarse partition can deliver the correct result.

\section{Problem Statement}\label{sec:problem}
Let $\Xvec=[X_1, X_2,\dots, X_N]$ be an $N$-dimensional RV with a probability measure $P_\Xvec$ absolutely continuous w.r.t. the Lebesgue measure $\mu$ ($P_\Xvec \ll \mu$). We require $P_\Xvec$ to be concentrated on $\dom{X}\subseteq\mathbb{R}^N$. This RV, which possesses a unique probability density function (PDF) $f_\Xvec$, is the input to the following multivariate, vector-valued function:
\begin{definition}\label{def:function}
Let $\gvec{:}\  \dom{X}\to\dom{Y}$, $\dom{X},\dom{Y}\subseteq\mathbb{R}^N$, be a surjective, Borel-measurable function defined in a piecewise manner:
\begin{equation}
 \gvec(\xvec) = \begin{cases}
             \gvec_1(\xvec), & \text{if } \xvec\in\dom{X}_1\\
             \gvec_2(\xvec), & \text{if } \xvec\in\dom{X}_2\\
							\vdots
% 						 \gvec_L(\xvec), & \text{if } \xvec\in\dom{X}_L
            \end{cases}
\end{equation}
where $\xvec=[x_1, x_2, \dots, x_N]$ and $\gvec_i{:}\ \dom{X}_i\to\dom{Y}_i$ bijectively\ifthenelse{\arxiv=1}{\footnote{In the univariate case, i.e., for $N=1$, this is equivalent to requiring that $g$ is piecewise strictly monotone.}}{}. Furthermore, let the Jacobian matrix $\Jac{\gvec}{\cdot}$ exist on the closures of $\dom{X}_i$. In addition to that, we require the Jacobian determinant, $|\det\Jac{\gvec}{\cdot}|$, to be non-zero $P_\Xvec$-almost everywhere.
\end{definition}

In accordance with previous work~\cite{Geiger_ISIT2011arXiv} the $\dom{X}_i$ are disjoint sets of positive $P_\Xvec$-measure which unite to $\dom{X}$, i.e., $\bigcup_{i} \dom{X}_i=\dom{X}$ and $\dom{X}_i \cap \dom{X}_j=\emptyset$ if $i\neq j$. Clearly, also the $\dom{Y}_i$ unite to $\dom{Y}$, but need not be disjoint. This definition ensures that the preimage $\preimV{\yvec}$ of each element $\yvec\in\dom{Y}$ is a countable set.

Using the method of transformation~\cite[pp.~244]{Papoulis_Probability} one obtains the PDF of the $N$-dimensional output RV $\Yvec=[Y_1,Y_2,\dots ,Y_N]$ as
\begin{equation}
 f_\Yvec(\yvec) = \sum_{\xvec_i\in\preimV{\yvec}} \frac{f_\Xvec(\xvec_i)}{|\det\Jac{\gvec}{\xvec_i}|}\label{eq:fy}
\end{equation}
where the sum is over all elements of the preimage. Note that since $\Yvec$ possesses a density, the corresponding probability measure $P_\Yvec$ is also absolutely continuous w.r.t. the Lebesgue measure.

\section{Main Results}\label{sec:mainResults}
We now state our main results:
\begin{thm}\label{thm:loss}
 The information loss induced by a function $\gvec$ satisfying Definition~\ref{def:function} is given as
\begin{equation}
 \ent{\Xvec|\Yvec} = \int_{\dom{X}} f_\Xvec(\xvec) \log\left( \frac{\sum_{\xvec_i\in\preimV{\gvec(\xvec)}} \frac{f_\Xvec(\xvec_i)}{|\det\Jac{\gvec}{\xvec_i}|}}{\frac{f_\Xvec(\xvec)}{|\det\Jac{\gvec}{\xvec}|}}\right) d\xvec.\label{eq:loss}
\end{equation}
\end{thm}

\ifthenelse{\arxiv=1}{
The proof of this Theorem can be found in the Appendix and, in a modified version for univariate functions, in~\cite{Geiger_ISIT2011arXiv}. Note that for univariate functions the Jacobian determinant is replaced by the derivative of the function.}{The proof of this Theorem can be found in the Appendix of~\cite{Geiger_ILStatic_arXiv} and, in a modified version for univariate functions, in~\cite{Geiger_ISIT2011arXiv}. Note that for univariate functions the Jacobian determinant is replaced by the derivative of the function.}

\begin{cor}\label{cor:diffLoss}
 The information loss induced by a function $\gvec$ satisfying Definition~\ref{def:function} is given as
\begin{equation}
 \ent{\Xvec|\Yvec} = \diffent{\Xvec}-\diffent{\Yvec} + \expec{\log|\det\Jac{\gvec}{\Xvec}|}
\end{equation}
\end{cor}

\begin{proof}
 The proof is obtained by recognizing the PDF of $\Yvec$ inside the logarithm in~\eqref{eq:loss} and by splitting the logarithm.
\end{proof}

This result is particularily interesting because it provides a link between information loss and differential entropies already anticipated in~\cite[pp.~660]{Papoulis_Probability}. There, it was claimed that
\begin{equation}
 \diffent{\Yvec} \leq \diffent{\Xvec} + \expec{\log|\det\Jac{\gvec}{\Xvec}|}\label{eq:papDiff}
\end{equation}
where equality holds iff $\gvec$ is bijective. While~\eqref{eq:papDiff} is actually another version of the DPI, Corollary~\ref{cor:diffLoss} quantifies how much information is lost by processing. In addition to that, a very similar expression denoted as \emph{folding entropy} has been presented in~\cite{Ruelle_EntropyProduction}, although in a completely different setting analyzing the entropy production of \emph{autonomous} dynamical systems.

We now introduce a discrete RV $W$ which depends on the set $\dom{X}_i$ from which $\Xvec$ was taken. In other words, for all $i$ we have $W=w_i$ iff $\xvec\in\dom{X}_i$. One can interpret this RV as being generated by a vector quantization of $\Xvec$ with a partition $\partit{}=\{\dom{X}_i\}$. With this new RV we can state

\begin{thm}\label{thm:WisX}
 The information loss is identical to the uncertainty about the set $\dom{X}_i$ from which the input was taken, i.e.,
\begin{equation}
 \ent{\Xvec|\Yvec} = \ent{W|\Yvec}.
\end{equation}
\end{thm}

The proof follows closely the proof provided in~\cite{Geiger_ISIT2011arXiv} and thus is omitted. However, this equivalence suggests a way of measuring information loss by means of proper quantization: Since $\ent{W|\Yvec}=\mutinf{W;\Xvec}-\mutinf{W;\Yvec}$ the loss can be determined by measuring mutual informations, which in this case are always finite (or, at least, bounded by $\ent{W}$). \ifthenelse{\arxiv=1}{In contrary to that, the mutual information in~\eqref{eq:limPart} of Theorem~\ref{thm:defEq} diverge to infinity; This expression was used in~\cite{Geiger_ISIT2011arXiv} for the information loss, highlighting the fact that both the self-information of $\Xvec$ and the information transfer from $\Xvec$ to $\Yvec$ are infinite.}{}
%In contrary to that, the mutual informations involved in Theorem~\ref{thm:defEq} diverge to infinity for the considered problem statement.

The interpretation derived from Theorem~\ref{thm:WisX} allows us now to provide upper bounds on the information loss:

\begin{thm}\label{thm:bounds}
 The information loss is upper bounded by
\begin{eqnarray}
 \ent{\Xvec|\Yvec} &\leq& \int_{\dom{Y}} f_\Yvec(\yvec) \log|\preimV{\yvec}| d\yvec\\
&\leq& \log\left(\sum_i\int_{\dom{Y}_i}f_\Yvec(\yvec)d\yvec\right)\\
&\leq& \max_{\yvec}\log|\preimV{\yvec}|.
\end{eqnarray}
\end{thm}

\begin{proof}
 We give here only a sketch of the proof: The first inequality results from bounding $\ent{W|\Yvec=\yvec}$ by the entropy of a uniform distribution on the preimage of $\yvec$. Jensen's inequality yields the second line of the Theorem. The coarsest bound is obtained by replacing the cardinality of the preimage by its maximal value.
\end{proof}

In this Theorem, we bounded the information loss given a certain output by the cardinality of the preimage. While the first bound considers the fact that the cardinality may actually depend on the output itself, the last bound incorporates the maximum cardinality only. In cases where the function from Definition~\ref{def:function} is defined not on a countable but on a finite number of subdomains this finite number can act as an upper bound (cf.~\cite{Geiger_ISIT2011arXiv}). Another straightforward upper bound, which is independent from the bounds in Theorem~\ref{thm:bounds} is obtained from Theorem~\ref{thm:WisX} by removing conditioning: 
\begin{equation}
 \ent{\Xvec|\Yvec}\leq\ent{W}=-\sum_i p_i\log p_i
\end{equation}
where $p_i=P_\Xvec(\dom{X}_i)=\int_{\dom{X}_i} f_\Xvec(\xvec)d\xvec$. It has to be noted, though, that depending on the function $\gvec$ all these bounds can be infinite while the information loss remains finite.

A further implication of introducing this discrete RV $W$ is that it allows us to perform investigations about reconstructing the input from the output. Currently, a Fano-type inequality bounding the reconstruction error by the information loss is under development. In addition to that, new upper bounds on the information loss related to the reconstruction error of optimal (in the \emph{maximum a posteriori} sense) and of simpler, sub-optimal estimators are analyzed.

\section{Functions with Infinite Information Loss}
\label{sec:extension}
We now drop the requirement of local bijectivity in Definition~\ref{def:function} to analyze a wider class of surjective, Borel-measurable functions $\gvec{:}\ \dom{X}\to\dom{Y}$. We keep the requirement that $P_\Xvec \ll \mu$ and thus $\Xvec$ possesses a density $f_\Xvec$ (positive on $\dom{X}$ and zero elsewhere). We maintain
\begin{thm}\label{thm:infLoss}
 Let $\gvec{:}\ \dom{X}\to\dom{Y}$ be a Borel-measurable function and let the continuous RV $\Xvec$ be the input to this function. If there exists a set $B\subseteq\dom{Y}$ of positive $P_\Yvec$-measure such that the preimage $\preimV{\yvec}$ is uncountable for every $\yvec\in B$, then the information loss is infinite.
\end{thm}

\begin{proof}
We notice that since $B\subseteq\dom{Y}$
\begin{IEEEeqnarray}{RCL}
 \ent{\Xvec|\Yvec}\ifthenelse{\arxiv=1}{&=&\int_{\dom{Y}}\ent{\Xvec|\Yvec=\yvec}dP_\Yvec(\yvec)\\}{}
&\geq& \int_{B}\ent{\Xvec|\Yvec=\yvec}dP_\Yvec(\yvec)
\end{IEEEeqnarray}
where the integral\ifthenelse{\arxiv=1}{s are}{ is} now written as Lebesgue integral\ifthenelse{\arxiv=1}{s}{}, since $P_\Yvec$ now not necessarily possesses a density.

Since on $B$ the preimage of every element is uncountable, we obtain with~\cite{Pinsker_InfoEngl} and the references therein $ \ent{\Xvec|\Yvec=\yvec}=\infty$ for all $\yvec\in B$, and, thus, $\ent{\Xvec|\Yvec}=\infty$.
\end{proof}

Note that the requirement of $B$ being a set of positive $P_\Yvec$-measure cannot be dropped, \ifthenelse{\arxiv=1}{as Example 4 in Section~\ref{sec:examples} illustrates.}{as one of the examples provided in~\cite{Geiger_ILStatic_arXiv} illustrates.} We immediately obtain the following

\begin{cor}
 Let $\gvec{:}\ \dom{X}\to\dom{Y}$ be a Borel-measurable function and let the continuous RV $\Xvec$ be the input to this function. If the probability measure of the output, $P_\Yvec$, possesses a non-vanishing discrete component, the information loss is infinite.
\end{cor}

\begin{proof}
 According to the Lebesgue-Radon-Nikodym theorem~\cite[pp.~121]{Rudin_Analysis3} every measure can be decomposed in a component absolutely continuous w.r.t. $\mu$ and a component singular to $\mu$. The latter part can further be decomposed into a singular continuous and a discrete part, where the latter places positive $P_\Yvec$-mass on points. Let $\yvec^*$ be such a point, i.e., $P_\Yvec(\yvec^*)>0$. As an immediate consequence, $P_\Xvec(\preimV{\yvec^*})>0$, which is only possible if $\preimV{\yvec^*}$ is uncountable ($P_\Xvec \ll \mu$).
\end{proof}

This result is also in accordance with intuition, as the analysis of a simple quantizer shows: While the entropy of the input RV is infinite ($\mutinf{\hat{\Xvec};\Xvec}\to\infty$ for $\hat{\Xvec}\to\Xvec$; cf.~\cite[pp.~654]{Papoulis_Probability}), the quantized output can contain only a finite amount of information ($\mutinf{\hat{\Xvec};\Yvec}\to\ent{\Yvec}<\infty$). In addition to that, the preimage of each possible output value $\yvec$ is a set of positive $P_\Xvec$-measure. The loss, as a consequence, is infinite.

While for the quantizer the preimage of each possible output value is a set of positive measure, there certainly are functions for which some outputs have a countable preimage and some whose preimage is a non-null set. An example of such a system is the limiter~\cite[Ex.~5-4]{Papoulis_Probability}. For such systems it can be shown that both the information loss $L(\Xvec\to \Yvec)=\ent{\Xvec|\Yvec}$ and the information transfer $\mutinf{\Xvec;\Yvec}$ are infinite.

Finally, there exist functions $\gvec$ for which the preimages of all output values $\yvec$ are null sets, but which still fulfill the conditions of Theorem~\ref{thm:infLoss}. Functions which project $\dom{X}$ on a lower-dimensional subspace of $\mathbb{R}^N$ fall into that category.

\ifthenelse{\arxiv=1}{
\section{Examples}
\label{sec:examples}

In this Section we illustrate our theoretical results with the help of examples. The logarithm is taken to base 2 unless otherwise noted.
\subsection{Example 1: A two-dimensional transform with finite information loss}
Let $\Xvec$ be uniformly distributed on the square $\dom{X}=[-a,a]\times [-a,a]$. Equivalently, the two constituing RVs $X_1$ and $X_2$ are independent and uniformly distributed on $[-a,a]$. In other words, while $f_\Xvec(\xvec)=1/4a^2$ for all $\xvec\in\dom{X}$, we have $f_X(x_i)=1/2a$ for $x_i\in[-a,a]$ and $i=1,2$.

We consider a function $\gvec$ performing the mapping:
\begin{eqnarray}
 Y_1&=&X_1\\ Y_2 &=& |X_1-X_2|
\end{eqnarray}
The corresponding Jacobian matrix is a triangular matrix
\begin{equation}
 \Jac{\gvec}{\xvec} = \left[\begin{array}{cc}
1 & 0 \\ \sgn{x_1-x_2} & \sgn{x_2-x_1}
\end{array}\right]
\end{equation}
where $\sgn{\cdot}$ is the sign-function. From this immediately follows that the magnitude of the determinant of the Jacobian matrix is unity for all possible values of $\Xvec$, i.e., $|\det\Jac{\gvec}{\xvec}|=1$ for all $\xvec\in\dom{X}$. The subsets of $\dom{X}$ on which the partitioned functions $\gvec_i$ are bijective are no intervals in this case; they are the triangular halves of the square induced by $x_1=x_2$ (see Fig.~\ref{fig:domainsEx1})
\begin{eqnarray}
 \dom{X}_1 &=& \{[x_1,x_2]\in\dom{X}:\ x_1> x_2\}\\
 \dom{X}_2 &=& \{[x_1,x_2]\in\dom{X}:\ x_1\leq x_2\}.
\end{eqnarray}
The preimage of $\gvec(\xvec)$ is, in any case,
\begin{equation}
 \{[x_1,x_2],[x_1,2x_1-x_2]\}\cap \dom{X}.
\end{equation}
The transform $\gvec$ is bijective whenever $[x_1,2x_1-x_2]\notin\dom{X}$, i.e., if $|2x_1-x_2|>a$.

With the PDF of $\Xvec$ and of its components we obtain for the information loss
\begin{equation}
 \ent{\Xvec|\Yvec} = \int_{-a}^a\int_{-a}^a \frac{1}{4a^2} \log \left(\frac{\frac{1}{2a}+f_X(2x_1-x_2)}{\frac{1}{2a}}\right) dx_1 dx_2
\end{equation}
which is non-zero only if $-a \leq 2x_1-x_2 \leq a$ (numerator and denominator cancel otherwise; no loss occurs in the bijective domain of the function). As a consequence,
\begin{eqnarray}
 \ent{\Xvec|\Yvec} &=& \int_{-a}^a \int_{\frac{x_2-a}{2}}^{\frac{x_2+a}{2}} \frac{\log 2}{4a^2} dx_1dx_2\\
 &=& \int_{-a}^a \frac{1}{4a} dx_2 = \frac{1}{2}.
\end{eqnarray}
The information loss is identical to a half bit. This is intuitive when looking at Fig.~\ref{fig:domainsEx1}, where it can be seen that any information loss occurs only on one half of the domain $\dom{X}$ (shaded in stronger colors). By destroying the sign information, in this area the information loss is equal to one bit.
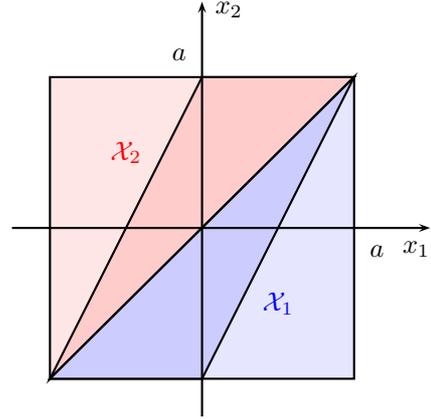
\begin{figure}
  \begin{center}
 \begin{pspicture}[showgrid=false](-3,-2.5)(3,3)
	\pspolygon[fillstyle=solid,fillcolor=red!10](-2,-2)(-2,2)(2,2)
	\pspolygon[fillstyle=solid,fillcolor=red!20](-2,-2)(0,2)(2,2)
	\pspolygon[fillstyle=solid,fillcolor=blue!10](-2,-2)(2,-2)(2,2)
	\pspolygon[fillstyle=solid,fillcolor=blue!20](-2,-2)(0,-2)(2,2)
	\psaxeslabels{->}(0,0)(-2.5,-2.5)(3,3){$x_1$}{$x_2$}
	\rput(2.3,-0.3){$a$} \rput(-0.3,2.3){$a$}
	\rput(-1,1){\textcolor{red}{$\dom{X}_2$}}\rput(1,-1){\textcolor{blue}{$\dom{X}_1$}}
 \end{pspicture}
\end{center}
\caption{Subdomains of Example 1. The partitioned functions $\gvec_i$ restricted to the domain of either color are bijective. Furthermore, the overall function $\gvec$ is bijective in areas with light shading.}
\label{fig:domainsEx1}
\end{figure}

\subsection{Example 2: Squaring a Gaussian RV}
Let $X$ be a zero-mean Gaussian RV with unit variance and differential entropy $\diffent{X}=\frac{1}{2}\ln(2\pi e)$ measured in nats. We consider the square of this RV, $Y=g(X)=X^2$, to illustrate the connection between information loss and differential entropy. The square of a Gaussian RV, $Y$, is $\chi^2$-distributed with one degree of freedom. Thus, the differential entropy of $Y$ is given by~\cite{Lazo_Entropies} \begin{IEEEeqnarray}{RCL}
 \diffent{Y} &=& \frac{1}{2}+\ln\left(2\Gamma\left(\frac{1}{2}\right)\right)+ \frac{1}{2}\psi\left(\frac{1}{2}\right)\\&=&\frac{1}{2}+\frac{1}{2}\ln\pi-\frac{\gamma}{2}
\end{IEEEeqnarray}
where $\Gamma(\cdot)$ and $\psi(\cdot)$ are the gamma- and digamma-functions~\cite[Ch.~6]{Abramowitz_Handbook} and $\gamma$ is the Euler-Mascheroni constant~\cite[pp.~3]{Abramowitz_Handbook}. With some calculus we obtain for the expected value of the derivative (taking the place of the Jacobian determinant in the univariate case)
\begin{equation}
 \expec{\ln|2x|} = \frac{1}{2}\ln 2 - \frac{\gamma}{2}.
\end{equation}
Subtracting differential entropies and adding the expected value of the derivative yields the information loss
\begin{eqnarray}
 \ent{X|Y} &=& \diffent{X}-\diffent{Y}+\expec{\ln|2x|}\\
&=& \frac{1}{2}\ln(2\pi e) - \frac{1}{2} - \frac{1}{2}\ln(\pi)  + \frac{1}{2}\ln 2 \\
&=& \ln2
\end{eqnarray}
again measured in nats. Changing the base of the logarithm to 2 we obtain an information loss of one bit. This is in perfect accordance with a previous result showing that the information loss of a square-law device is equal to one bit if the PDF of the input has even symmetry~\cite{Geiger_ISIT2011arXiv}.

\subsection{Example 3: Exponential RV and infinite bounds}
In this example we consider an exponential input with PDF
\begin{equation}
 f_X(x)=\lambda\e{-\lambda x}
\end{equation}
and a piecewise linear function
\begin{equation}
 g(x)=x-\frac{\lfloor \lambda x \rfloor}{\lambda}.
\end{equation}
The PDF and the function are depicted in Fig.~\ref{fig:functionsEx3}.

\begin{figure}
  \begin{center}
 \begin{pspicture}[showgrid=false](-1,-3.5)(5,2.5)
	\psaxeslabels{->}(0,0)(-0.5,-0.5)(5,2){$x$}{$f_X(x)$}
  \psplot[style=Graph,linecolor=red,plotpoints=500]{0}{4.5}{2.718 x 1.5 mul neg exp 1.5 mul}
	\psTick{0}(0,1.5) \rput(-0.5,1.5){$\lambda$}
	\psaxeslabels{->}(0,-3)(-0.5,-3.5)(5,-1){$x$}{$g(x)$}
  \psplot[style=Graph,linecolor=black,plotpoints=500]{0}{1}{x 3 sub}
	\psplot[style=Graph,linecolor=black,plotpoints=500]{1}{2}{x 4 sub}
	\psplot[style=Graph,linecolor=black,plotpoints=500]{2}{3}{x 5 sub}
	\psplot[style=Graph,linecolor=black,plotpoints=500]{3}{4}{x 6 sub}
	\psplot[style=Graph,linecolor=black,plotpoints=500]{4}{4.5}{x 7 sub}
	\psTick{0}(0,-2) \rput(-0.5,-2){$\frac{1}{\lambda}$}
	\psTick{90}(1,-3) \rput(1,-3.5){$\frac{1}{\lambda}$}
 \end{pspicture}
\end{center}
\caption{PDF $f_X$ and piecewise linear function $g$ of Example 3.}
\label{fig:functionsEx3}
\end{figure}
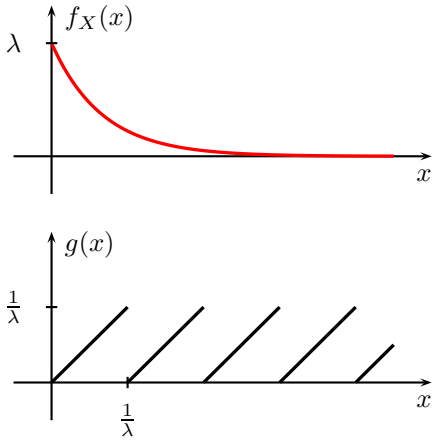

We obviously have $\dom{X}=[0,\infty)$ and $\dom{Y}=[0,\frac{1}{\lambda})$, while $g$ partitions $\dom{X}$ in a countable number of intervals of length $\frac{1}{\lambda}$. In other words,
\begin{equation}
 \dom{X}_k = \left[\frac{k-1}{\lambda},\frac{k}{\lambda} \right)
\end{equation}
and $g(\dom{X}_k)=\dom{Y}$ for all $k=1,2,\dots$. From this follows that for every $y\in\dom{Y}$ the preimage contains an element from each subdomain $\dom{X}_k$; thus, the bounds from Theorem~\ref{thm:bounds} all evaluate to $\ent{X|Y}\leq\infty$. However, it can be shown that the other bound, $\ent{X|Y}\leq\ent{W}$ is tight in this case:
With 
\begin{equation}
 p_k=P_X(\dom{X}_k)=\int_{\dom{X}_k}f_X(x)dx = (1-\e{-1})\e{-k+1}
\end{equation}
we obtain $\ent{W}=-\log(1-\e{-1})+\frac{\e{-1}}{1-\e{-1}}\approx 1.24$. The same result is obtained for a direct evaluation of Theorem~\ref{thm:loss}.

\subsection{Example 4: An almost invertible transform with zero information loss}
As a next example consider a two-dimensional RV $\Xvec$ which places probability mass uniformly on the unit disc, i.e.,
\begin{equation}
 f_\Xvec(\xvec)=\begin{cases}
                 \frac{1}{\pi}, &\text{ if } ||\xvec||\leq 1\\ 0, &\text{ else}
                \end{cases}
\end{equation}
where $||\cdot||$ is the Euclidean norm. Thus, $\dom{X}=\{\xvec\in\mathbb{R}^2: ||\xvec||\leq 1\}$. The cartesian coordinates $\xvec$ are now transformed to polar coordinates in a special way, namely:
\begin{IEEEeqnarray}{RCL}
 y_1&=&\begin{cases}
      ||\xvec||, &\text{ if }||\xvec||<1\\ 0,&\text{ else}     \end{cases}\\
y_2&=&\begin{cases}
       \arctan(\frac{x_2}{x_1})+\pi(1-\sgn{x_1}),&\text{ if } 0<||\xvec||<1\\
0,&\text{ else}
      \end{cases}\notag\\
\end{IEEEeqnarray}
This mapping together with the domains of $\Xvec$ and $\Yvec$ is illustrated in Fig.~\ref{fig:mappingEx4} (left and upper right diagram).
\begin{figure}
  \begin{center}
 \begin{pspicture}[showgrid=false](-4,-3.5)(4,2.5)
  \psring[fillcolor=red!10](-2,0){0}{1.5}
	\pscircle[linecolor=red](-2,0){1.5}
	\psaxeslabels{->}(-2,0)(-3.8,-1.8)(-0.2,1.8){$x_1$}{$x_2$}
  \psTick{0}(-2,1.5) \rput(-2.2,1.7){$1$}
	
	\psframe*[linecolor=red!10,](1.5,-1.)(3.5,2)
	\psaxeslabels{->}(1.5,-1)(1.2,-1.3)(4,2.3){$y_1$}{$y_2$}
	\psring[fillcolor=red](1.5,-1){0}{0.1}
	\psTick{0}(1.5,2) \rput(1.2,2){$2\pi$}
	\psTick{90}(3.5,-1) \rput(3.5,-1.3){$1$}

	\psaxeslabels{->}(1.5,-2.5)(1.2,-2.8)(4,-1.8){$y_1$}{$y_2$}
	\psline[linecolor=red,linewidth=1pt](1.5,-2.5)(3.5,-2.5)
	\psTick{90}(3.5,-2.5) \rput(3.5,-2.8){$1$}
 \end{pspicture}
\end{center}
\caption{Mapping of domains in Examples 4 and 5. The solid red circle in the left diagram and the red dot in the upper right diagram correspond to each other, illustrating the mapping of an uncountable $P_\Xvec$-null set to a point. The lightly shaded areas are mapped bijectively in Example 4. In Example 5, the disc in the left diagram is mapped to the solid red line in the lower right diagram.}
\label{fig:mappingEx4}
\end{figure}
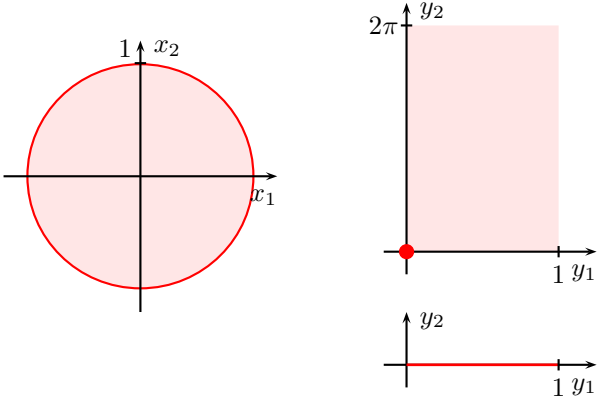

As a direct consequence we have $\dom{Y}=(0,1) \times [0,2\pi)\cup\{0,0\}$. Observe that not only the point $\xvec=\{0,0\}$ is mapped to the point $\yvec=\{0,0\}$, but that also the unit circle $\dom{S}=\{\xvec:||\xvec||=1\}$ is mapped to $\yvec=\{0,0\}$. As a consequence, the preimage of $\{0,0\}$ under $\gvec$ is uncountable. However, since a circle in $\mathbb{R}^2$ is a Lebesgue null-set and thus $P_\Xvec(\dom{S})=0$, also $P_\Yvec(\{0,0\})=0$ and the conditions of Theorem~\ref{thm:infLoss} are not met. Indeed, since $\ent{\Xvec|\Yvec=\yvec}=0$ $P_\Yvec$-almost everywhere, it can be shown that $\ent{\Xvec|\Yvec}=0$.

\subsection{Example 5: A mapping to a subspace of lower dimensionality}
Consider again a uniform distribution on the unit disc, as it was used in Example 4. Now, however, let $\gvec$ be such that only the radius is computed while the angle is lost, i.e.,
\begin{IEEEeqnarray}{RCL}
 y_1&=& ||\xvec||\\
y_2&=&0.
\end{IEEEeqnarray}
Note that here only the origin $\{0,0\}$ is mapped bijectively, while for all other $\yvec\in\dom{Y}=[0,1]\times \{0\}$ the preimage under $\gvec$ is uncountable (a circle around the origin with radius $y_1$). Indeed, in this particular example, the probability measure $P_\Yvec$ is \emph{not} discrete, but singular continuous: Each point has zero $P_\Yvec$-measure (circles are Lebesgue null-sets), but $P_\Yvec$ is not absolutely continuous w.r.t. the two-dimensional Lebesgue measure $\mu$. Clearly, $\mu(\dom{Y})=0$ while $P_\Yvec(\dom{Y})=1$. Since the preimage is uncountable on a set of positive $P_\Yvec$-measure, we have $\ent{\Xvec|\Yvec}=\infty$.

\subsection{Example 6: Another two-dimensional transform with finite information loss}
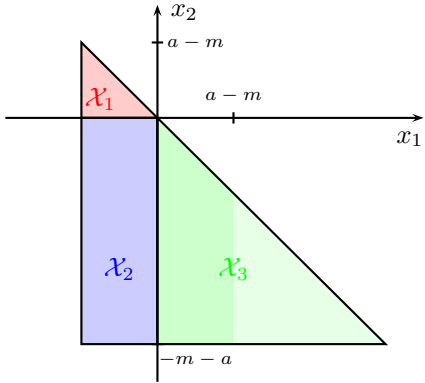
\begin{figure}
  \begin{center}
 \begin{pspicture}[showgrid=false](-3,-3)(3,3)
	\pspolygon[fillstyle=solid,fillcolor=red!20](-2,2)(-2,1)(-1,1)
	\pspolygon[fillstyle=solid,fillcolor=blue!20](-2,1)(-1,1)(-1,-2)(-2,-2)
	\pspolygon*[linecolor=green!10](-1,1)(-1,-2)(2,-2)
	\pspolygon*[linecolor=green!20](-1,1)(-1,-2)(0,-2)(0,0)
 	\pspolygon(-1,1)(-1,-2)(2,-2)
	\psaxeslabels{->}(-1,1)(-3,-2.5)(2.5,2.5){$x_1$}{$x_2$}
	\psTick{0}(-1,2) \rput(-0.5,2){\scriptsize $a-m$}
	\psTick{90}(0,1) \rput(0,1.3){\scriptsize $a-m$}
	\psTick{0}(-1,-2) \rput(-0.5,-2.2){\scriptsize $-m-a$}
	\rput(-1.75,1.25){\textcolor{red}{$\dom{X}_1$}}\rput(-1.5,-1){\textcolor{blue}{$\dom{X}_2$}}\rput(0,-1){\textcolor{green}{$\dom{X}_3$}}
 \end{pspicture}
\end{center}
\caption{Subdomains of Example 6. The functions $\gvec_i$ restricted to a domain of either color are bijective. Furthermore, the overall function $\gvec$ is bijective in areas with light shading.}
\label{fig:domainsEx6}
\end{figure}
Finally, consider a uniform distribution on a triangle defined by
\begin{equation}
 \dom{X}=\{\xvec\in\mathbb{R}^2: x_1\in[m-a,m+a],x_2\in[-m-a,-x_1]\}
\end{equation}
where $0\leq m \leq a$ and $a>0$. Thus, the PDF of $\Xvec$ is given as $f_\Xvec(\xvec)=\frac{1}{2a^2}$ if $\xvec\in\dom{X}$ and zero elsewhere (see Fig.~\ref{fig:domainsEx6}). The function $\gvec$ takes the magnitude of each coordinate, i.e., $y_i=|x_i|$, where $i=1,2$. We now try to derive the information loss as a function of $m$.

First we can identify three subsets of $\dom{X}$ which are mapped bijectively by restricting $\gvec$ to these sets, namely $\dom{X}_1=\{\xvec\in\dom{X}: x_1\leq 0, x_2\geq 0\}$, $\dom{X}_2=\{\xvec\in\dom{X}: x_1\leq 0, x_2<0\}$, and $\dom{X}_3=\{\xvec\in\dom{X}:x_1>0,x_2<0\}$. Furthermore, for $m\geq 0$ a part of $\dom{X}_3$ is mapped bijectively by $\gvec$ (lighter shading in Fig.~\ref{fig:domainsEx6}). The probability mass contained in this subset $\dom{X}_b$ can be shown to equal $P_b=P_\Xvec(\dom{X}_b)=\frac{m^2}{a^2}$. For all other possible input values $\xvec$ the preimage of $\gvec(\xvec)$ has exactly two elements: One of them is located in $\dom{X}_2$, the other either in $\dom{X}_1$ or in $\dom{X}_3\setminus\dom{X}_b$. Due to the uniformity of $\Xvec$ and since the Jacobian determinant is identical to unity for all $\xvec\in\dom{X}$ both of these preimages are equally likely. Thus, on $\dom{X}\setminus\dom{X}_b$ the information loss is identical to one bit. In other words,
\begin{equation}
 \ent{\Xvec|\Yvec=\yvec}=1
\end{equation}
for all $\yvec\in\gvec(\dom{X}\setminus\dom{X}_b)$. We therefore obtain with $P_b=\frac{m^2}{a^2}$ an information loss equal to $\ent{\Xvec|\Yvec}=1-\frac{m^2}{a^2}$.

From the probability masses contained in the sets $\dom{X}_1$, $\dom{X}_2$, and $\dom{X}_3$ we can compute an upper bound on the information loss:
\begin{equation}
 \ent{W}=\frac{m^2}{2a^2}+\frac{3}{2}-\log\frac{a^2-m^2}{a^2}+\frac{m}{a}\log\frac{a-m}{a+m}.
\end{equation}
And evaluating the bounds of Theorem~\ref{thm:bounds} yields
\begin{equation}
 \ent{\Xvec|\Yvec}\leq 1-\frac{m^2}{a^2}\leq \log(2-\frac{m^2}{a^2})\leq 1
\end{equation}
which for $m=0$ all reduce to one bit. In particular, it can be seen that in this case the smallest bound of Theorem~\ref{thm:bounds} is exact.

The exact information loss, together with the second smallest bound from Theorem~\ref{thm:bounds} and with the bound from $\ent{W}$, is shown in Fig.~\ref{fig:lossEx6}. As it can be seen, the closer the parameter $m$ approaches $a$, the smaller the information loss gets. Conversely, for $m=0$ the information loss is exactly one bit. Moreover, it turns out that the bound from $\ent{W}$ is rather loose in this case.

\begin{figure}
 \centering
 \includegraphics[width=0.5\textwidth]{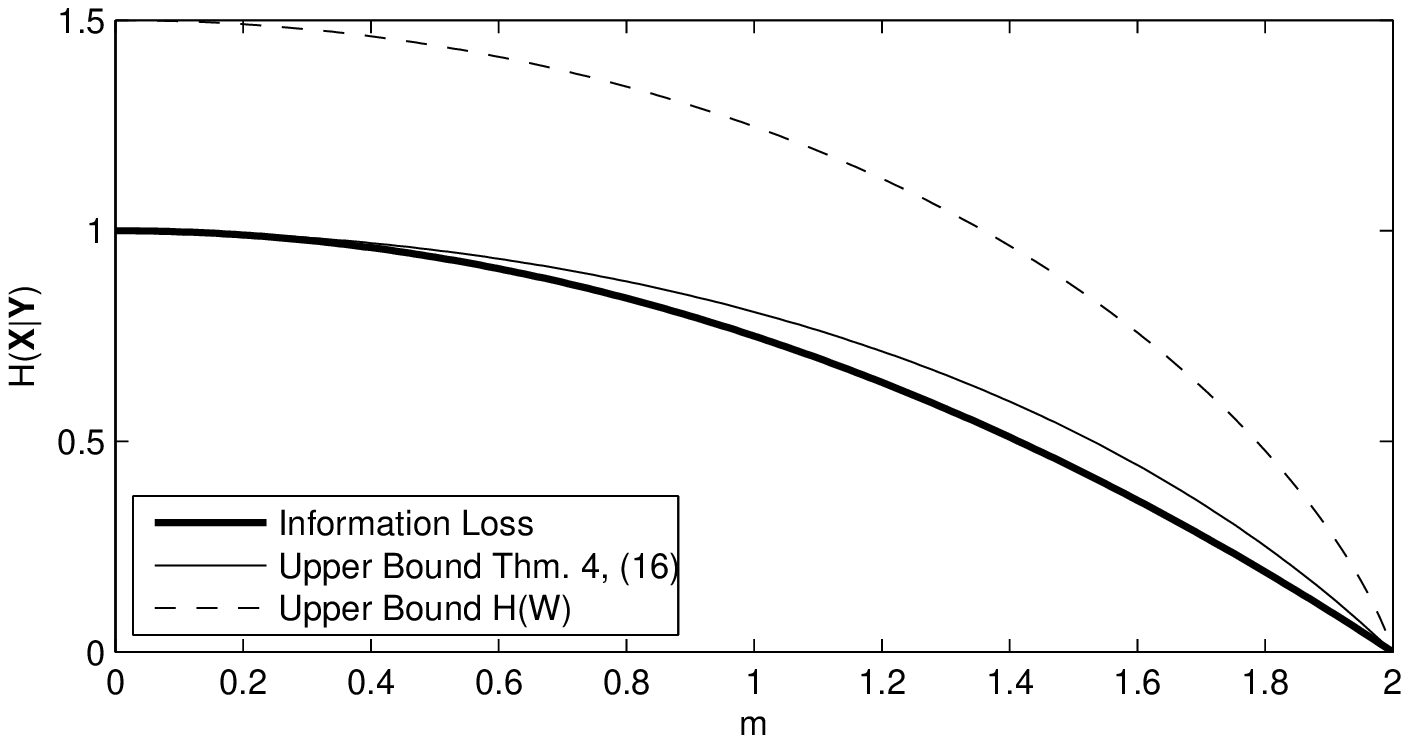}
 \caption{Information loss $\ent{\Xvec|\Yvec}$ of Example 6 for $a=2$.}
 \label{fig:lossEx6}
\end{figure}

}{\section{Example: A two-dimensional transform with finite information loss}
\label{sec:examples}
In this Section we illustrate our theoretical results with the help of an example. The logarithm is taken to base 2.
% \subsection{Example 1: A two-dimensional transform with finite information loss}
Let $\Xvec$ be uniformly distributed on the square $\dom{X}=[-a,a]\times [-a,a]$. Equivalently, the two constituing RVs $X_1$ and $X_2$ are independent and uniformly distributed on $[-a,a]$. In other words, while $f_\Xvec(\xvec)=1/4a^2$ for all $\xvec\in\dom{X}$, we have $f_X(x_i)=1/2a$ for $x_i\in[-a,a]$ and $i=1,2$.

We consider a function $\gvec$ performing the mapping:
\begin{eqnarray}
 Y_1&=&X_1\\ Y_2 &=& |X_1-X_2|
\end{eqnarray}
The corresponding Jacobian matrix is a triangular matrix
\begin{equation}
 \Jac{\gvec}{\xvec} = \left[\begin{array}{cc}
1 & 0 \\ \sgn{x_1-x_2} & \sgn{x_2-x_1}
\end{array}\right]
\end{equation}
where $\sgn{\cdot}$ is the sign-function. From this immediately follows that the magnitude of the determinant of the Jacobian matrix is unity for all possible values of $\Xvec$, i.e., $|\det\Jac{\gvec}{\xvec}|=1$ for all $\xvec\in\dom{X}$. The subsets of $\dom{X}$ on which the partitioned functions $\gvec_i$ are bijective are no intervals in this case; they are the triangular halves of the square induced by $x_1=x_2$ (see Fig.~\ref{fig:domainsEx1})
\begin{eqnarray}
 \dom{X}_1 &=& \{[x_1,x_2]\in\dom{X}:\ x_1> x_2\}\\
 \dom{X}_2 &=& \{[x_1,x_2]\in\dom{X}:\ x_1\leq x_2\}.
\end{eqnarray}
The preimage of $\gvec(\xvec)$ is, in any case,
\begin{equation}
 \{[x_1,x_2],[x_1,2x_1-x_2]\}\cap \dom{X}.
\end{equation}
The transform $\gvec$ is bijective whenever $[x_1,2x_1-x_2]\notin\dom{X}$, i.e., if $|2x_1-x_2|>a$.

With the PDF of $\Xvec$ and of its components we obtain for the information loss
\begin{equation}
 \ent{\Xvec|\Yvec} = \int_{-a}^a\int_{-a}^a \frac{1}{4a^2} \log \left(\frac{\frac{1}{2a}+f_X(2x_1-x_2)}{\frac{1}{2a}}\right) dx_1 dx_2
\end{equation}
which is non-zero only if $-a \leq 2x_1-x_2 \leq a$ (numerator and denominator cancel otherwise; no loss occurs in the bijective domain of the function). As a consequence,
\begin{equation}
 \ent{\Xvec|\Yvec} = \int_{-a}^a \int_{\frac{x_2-a}{2}}^{\frac{x_2+a}{2}} \frac{\log 2}{4a^2} dx_1dx_2%\\
 = \int_{-a}^a \frac{1}{4a} dx_2 = \frac{1}{2}.
\end{equation}
The information loss is identical to a half bit. This is intuitive by looking at Fig.~\ref{fig:domainsEx1}, where it can be seen that any information loss occurs only on one half of the domain $\dom{X}$ (shaded in stronger colors). By destroying the sign information, in this area the information loss is equal to one bit.
\begin{figure}
  \begin{center}
 \begin{pspicture}[showgrid=false](-3,-2.2)(3,3,1)
	\pspolygon[fillstyle=solid,fillcolor=red!10](-2,-2)(-2,2)(2,2)
	\pspolygon[fillstyle=solid,fillcolor=red!20](-2,-2)(0,2)(2,2)
	\pspolygon[fillstyle=solid,fillcolor=blue!10](-2,-2)(2,-2)(2,2)
	\pspolygon[fillstyle=solid,fillcolor=blue!20](-2,-2)(0,-2)(2,2)
	\psaxeslabels{->}(0,0)(-2.5,-2.5)(3,2.7){$x_1$}{$x_2$}
	\rput(2.3,-0.3){$a$} \rput(-0.3,2.3){$a$}
	\rput(-1,1){\textcolor{red}{$\dom{X}_2$}}\rput(1,-1){\textcolor{blue}{$\dom{X}_1$}}
 \end{pspicture}
\end{center}
\caption{Subdomains of Example 1. The partitioned functions $\gvec_i$ restricted to the domain of either color are bijective. Furthermore, the overall function $\gvec$ is bijective in areas with light shading.}
\label{fig:domainsEx1}
\end{figure}
}

\section{Conclusion}
In this work, we proposed a mathematically concise definition of information loss for the purpose of establishing a system theory from an information-theoretic point of view. For a certain class of multivariate, vector-valued functions and continuous input variables this information loss was quantified, and the result is accompanied by convenient upper bounds. We further showed a connection between information loss and the differential entropies of the input and output variables.

Finally, a class of systems has been identified for which the information loss is infinite. Vector-quantizers and limiters belong to that class, but also functions which project the input space onto a space of lower dimensionality.

\ifthenelse{\arxiv=1}{
\appendix
\section*{Proof of Theorem~\ref{thm:loss}}
For the proof we use~\eqref{eq:limPart} of Theorem~\ref{thm:loss}, where we take the limit of a sequence of increasingly fine partitions $\partit{n}=\{\hat{\dom{X}}_k^{(n)}\}$ satisfying~\eqref{eq:partLim}. For a given $n$ we write the resulting mutual information $\mutinf{\hat{\Xvec}_n;\Xvec}$ as
\begin{equation}
 \mutinf{\hat{\Xvec}_n;\Xvec} = \expec{\kld{f_{\Xvec|\hat{\Xvec}_n}(\cdot,\hat{\xvec})}{f_\Xvec(\cdot)}}
\end{equation}
where $\kld{\cdot}{\cdot}$ denotes the Kullback-Leibler divergence and the expection is w.r.t. $\hat{\Xvec}_n$. Note that for each possible outcome $\hat{\xvec}_k$ of $\hat{\Xvec}_n$ the conditional probability measure $P_{\Xvec|\hat{\xvec}_k}$ is absolutely continuous w.r.t. the Lebesgue measure (cf.~Section~\ref{sec:extension}). It thus possesses a density
\begin{equation}
 f_{\Xvec|\hat{\Xvec}_n}(\xvec,\hat{\xvec}_k) = \begin{cases} \frac{f_\Xvec(\xvec)}{p(\hat{\xvec}_k)}, & \text{ if } \xvec\in\hat{\dom{X}}_k^{(n)}\\
0,& \text{ else }
\end{cases}\label{eq:fXXn}
\end{equation}
where $p(\hat{\xvec}_k)=P_\Xvec(\hat{\dom{X}}_k^{(n)})$. With the definition of the Kullback-Leibler divergence~\cite[Lemma~5.2.3]{Gray_Entropy} and~\cite[Thm.~5-1]{Papoulis_Probability} we can write the difference of mutual informations in Theorem~\ref{thm:defEq} as
\begin{multline}
 \mutinf{\hat{\Xvec}_n;\Xvec}-\mutinf{\hat{\Xvec}_n;\Yvec} =\\ \sum_k p(\hat{\xvec}_k) \int_{\hat{\dom{X}}_k^{(n)}} \frac{f_\Xvec(\xvec)}{p(\hat{\xvec}_k)} \log \left(\frac{f_{\Xvec|\hat{\Xvec}_n}(\xvec,\hat{\xvec}_k)f_\Yvec (\gvec(\xvec))}{f_{\Yvec|\hat{\Xvec}_n}(\gvec(\xvec),\hat{\xvec}_k)f_\Xvec(\xvec)}\right) d\xvec.
\end{multline}
Rewriting with the indicator function
\begin{equation}
 \eye_A(x) = \begin{cases}
              1, &\text{ if } x\in A\\ 0,&\text{ else}
             \end{cases}
\end{equation}
 this yields
\begin{multline*}
 \mutinf{\hat{\Xvec}_n;\Xvec}-\mutinf{\hat{\Xvec}_n;\Yvec} =\\  \int\limits_{\dom{X}} f_\Xvec(\xvec)\sum_k\left(\eye_{\hat{\dom{X}}_k^{(n)}}(\xvec)\log \left(\frac{f_{\Xvec|\hat{\Xvec}_n}(\xvec,\hat{\xvec}_k)f_\Yvec (\gvec(\xvec))}{f_{\Yvec|\hat{\Xvec}_n}(\gvec(\xvec),\hat{\xvec}_k)f_\Xvec(\xvec)}\right)\right) d\xvec.
\end{multline*}
We can now exploit the relationship~\eqref{eq:fy} for the conditional PDF of $\Yvec$ given $\hat{\Xvec}_n$, and with~\eqref{eq:fXXn} we realize that the function under the integral is monotonically increasing in $n$: Indeed, for finer partitions it is less likely that any element of the preimage $\preimV{\gvec(\xvec)}$ other than $\xvec$ lies in $\hat{\dom{X}}_k^{(n)}$, thus $f_{\Yvec|\hat{\Xvec}_n}(\gvec(\xvec),\hat{\xvec}_k)$ converges to $\frac{f_{\Xvec|\hat{\Xvec}_n}(\xvec,\hat{\xvec}_k) }{|\det\Jac{\gvec}{\xvec}|}$. This holds for all $k$, thus, invoking the monotone converge theorem~\cite[pp.~21]{Rudin_Analysis3} and cancelling the conditional PDFs eliminates the dependence on $k$ and the sum over indicator functions ($\bigcup_k \hat{\dom{X}}_k^{(n)} = \dom{X}$). Substituting the PDF of $\Yvec$ with~\eqref{eq:fy} completes the proof.\endproof}{}

\bibliographystyle{IEEEtran}
\bibliography{IEEEabrv,/afs/spsc.tugraz.at/project/IT4SP/1_d/Papers/InformationProcessing.bib,%
/afs/spsc.tugraz.at/project/IT4SP/1_d/Papers/ProbabilityPapers.bib,%
/afs/spsc.tugraz.at/user/bgeiger/includes/textbooks.bib,%
/afs/spsc.tugraz.at/user/bgeiger/includes/myOwn.bib,%
/afs/spsc.tugraz.at/user/bgeiger/includes/UWB.bib,%
/afs/spsc.tugraz.at/project/IT4SP/1_d/Papers/InformationWaves.bib,%
/afs/spsc.tugraz.at/project/IT4SP/1_d/Papers/ITBasics.bib,%
/afs/spsc.tugraz.at/project/IT4SP/1_d/Papers/HMMRate.bib,%
/afs/spsc.tugraz.at/project/IT4SP/1_d/Papers/ITAlgos.bib}

\end{document}

%% file: abbrevations_processing.tex
% Signals...
\newcommand{\x}[1]{x[#1]}
\newcommand{\y}[1]{y[#1]}

% PDFs
\newcommand{\pdfy}{f_Y(y)}

% Entropies
% \newcommand{\ent}[1]{H\left(#1\right)}
\newcommand{\ent}[1]{H(#1)}
\newcommand{\diffent}[1]{h(#1)}
\newcommand{\derate}[1]{\bar{h}\left(\mathbf{#1}\right)}
\newcommand{\mutinf}[1]{I(#1)}
\newcommand{\ginf}[1]{I_G(#1)}
\newcommand{\kld}[2]{D(#1||#2)}
\newcommand{\binent}[1]{H_2(#1)}
\newcommand{\binentneg}[1]{H_2^{-1}\left(#1\right)}
\newcommand{\entrate}[1]{\bar{H}(\mathbf{#1})}
\newcommand{\mutrate}[1]{\mutinf{\mathbf{#1}}}
\newcommand{\redrate}[1]{\bar{R}(\mathbf{#1})}
\newcommand{\pinrate}[1]{\vec{I}(\mathbf{#1})}
\newcommand{\lossrate}[1]{L(\mathbf{#1})}

% Domains and Sets
\newcommand{\dom}[1]{\mathcal{#1}}
\newcommand{\indset}[1]{\mathbb{I}\left({#1}\right)}

% Distributions...
\newcommand{\unif}[2]{\mathcal{U}\left(#1,#2\right)}
\newcommand{\chis}[1]{\chi^2\left(#1\right)}
\newcommand{\chir}[1]{\chi\left(#1\right)}
\newcommand{\normdist}[2]{\mathcal{N}\left(#1,#2\right)}
\newcommand{\Prob}[1]{\mathrm{Pr}(#1)}
\newcommand{\Mar}[1]{\mathrm{Mar}(#1)}
\newcommand{\Qfunc}[1]{Q\left(#1\right)}

% Functions...
\newcommand{\expec}[1]{\mathrm{E}\left\{#1\right\}}
\newcommand{\expecwrt}[2]{\mathrm{E}_{#1}\left\{#2\right\}}
\newcommand{\var}[1]{\mathrm{Var}\left\{#1\right\}}
\renewcommand{\det}{\mathrm{det}}
\newcommand{\cov}[1]{\mathrm{Cov}\left\{#1\right\}}
\newcommand{\sgn}[1]{\mathrm{sgn}\left(#1\right)}
\newcommand{\sinc}[1]{\mathrm{sinc}\left(#1\right)}
\newcommand{\e}[1]{\mathrm{e}^{#1}}
\newcommand{\multint}{\iint{\cdots}\int}
\newcommand{\modd}[3]{((#1))_{#2}^{#3}}
\newcommand{\quant}[1]{Q\left(#1\right)}

% Vectors and Matrices
\newcommand{\ivec}{\mathbf{i}}
\newcommand{\hvec}{\mathbf{h}}
\newcommand{\gvec}{\mathbf{g}}
\newcommand{\avec}{\mathbf{a}}
\newcommand{\kvec}{\mathbf{k}}
\newcommand{\fvec}{\mathbf{f}}
\newcommand{\vvec}{\mathbf{v}}
\newcommand{\xvec}{\mathbf{x}}
\newcommand{\Xvec}{\mathbf{X}}
\newcommand{\Xhvec}{\hat{\mathbf{X}}}
\newcommand{\xhvec}{\hat{\mathbf{x}}}
\newcommand{\xtvec}{\tilde{\mathbf{x}}}
\newcommand{\Yvec}{\mathbf{Y}}
\newcommand{\yvec}{\mathbf{y}}
\newcommand{\Zvec}{\mathbf{Z}}
\newcommand{\wvec}{\mathbf{w}}
\newcommand{\Wvec}{\mathbf{W}}
\newcommand{\Hmat}{\mathbf{H}}
\newcommand{\Amat}{\mathbf{A}}
\newcommand{\Fmat}{\mathbf{F}}

\newcommand{\zerovec}{\mathbf{0}}
\newcommand{\eye}{\mathbf{I}}
\newcommand{\evec}{\mathbf{i}}

\newcommand{\zeroone}{\left[\begin{array}{c}\zerovec^T\\ \eye\end{array} \right]}
\newcommand{\zerooneT}{\left[\begin{array}{cc}\zerovec & \eye\end{array} \right]}
\newcommand{\zerooneM}{\left[\begin{array}{cc}\zerovec &\zerovec^T\\\zerovec& \eye\end{array} \right]}

\newcommand{\Cxx}{\mathbf{C}_{XX}}
\newcommand{\Cxh}{\mathbf{C}_{\hat{X}\hat{X}}}
\newcommand{\rxx}{\mathbf{r}_{XX}}
\newcommand{\Cxy}{\mathbf{C}_{XY}}
\newcommand{\Cyy}{\mathbf{C}_{YY}}
\newcommand{\Cnn}{\mathbf{C}_{NN}}
\newcommand{\Cyx}{\mathbf{C}_{YX}}
\newcommand{\Cygx}{\mathbf{C}_{Y|X}}

\newcommand{\Jac}[2]{\mathcal{J}_{#1}(#2)}

% Other stuff
\newcommand{\NN}{{N{\times}N}}
\newcommand{\perr}{P_e}
\newcommand{\perh}{\hat{\perr}}
\newcommand{\pert}{\tilde{\perr}}

% Index
% \newcommand{\vecind}[1]{\mathbf{#1}}
\newcommand{\vecind}[1]{#1_0^n}
\newcommand{\roots}[2]{{#1}_{#2}^{(i_{#2})}}
\newcommand{\rootx}[1]{x_{#1}^{(i)}}
\newcommand{\rootn}[2]{x_{#1}^{#2,(i)}}

% Abbrevations
\newcommand{\markkern}[1]{f_M(#1)}
\newcommand{\pole}{a_1}
\newcommand{\preim}[1]{g^{-1}[#1]}
\newcommand{\preimV}[1]{\mathbf{g}^{-1}[#1]}
\newcommand{\Xmax}{\bar{X}}
\newcommand{\Xmin}{\underbar{X}}
\newcommand{\xmax}{x_{\max}}
\newcommand{\xmin}{x_{\min}}
\newcommand{\limn}{\lim_{n\to\infty}}
\newcommand{\limX}{\lim_{\hat{\Xvec}\to\Xvec}}
\newcommand{\limXo}{\lim_{\hat{X}_1\to X_1}}
\newcommand{\sumin}{\sum_{i=1}^n}
\newcommand{\finv}{f_\mathrm{inv}}%f_{X_n}^{-1}
\newcommand{\ejtheta}{\e{\jmath\theta}}
\newcommand{\khat}{\bar{k}}
\newcommand{\modeq}[1]{g(#1)}
\newcommand{\partit}[1]{\mathcal{P}_{#1}}
\newcommand{\psd}[1]{S_{#1}(\e{\jmath \theta})}
\newcommand{\borel}[1]{\mathfrak{B}(#1)}

% signal blocks
\newcommand{\delay}[2]{\psblock(#1){#2}{\footnotesize$z^{-1}$}}
\newcommand{\Quant}[2]{\psblock(#1){#2}{\footnotesize$\quant{\cdot}$}}
\newcommand{\moddev}[2]{\psblock(#1){#2}{\footnotesize$\modeq{\cdot}$}}

%% file: IZS2011.bbl
\begin{thebibliography}{10}
\providecommand{\url}[1]{#1}
\csname url@samestyle\endcsname
\providecommand{\newblock}{\relax}
\providecommand{\bibinfo}[2]{#2}
\providecommand{\BIBentrySTDinterwordspacing}{\spaceskip=0pt\relax}
\providecommand{\BIBentryALTinterwordstretchfactor}{4}
\providecommand{\BIBentryALTinterwordspacing}{\spaceskip=\fontdimen2\font plus
\BIBentryALTinterwordstretchfactor\fontdimen3\font minus
  \fontdimen4\font\relax}
\providecommand{\BIBforeignlanguage}[2]{{%
\expandafter\ifx\csname l@#1\endcsname\relax
\typeout{** WARNING: IEEEtran.bst: No hyphenation pattern has been}%
\typeout{** loaded for the language `#1'. Using the pattern for}%
\typeout{** the default language instead.}%
\else
\language=\csname l@#1\endcsname
\fi
#2}}
\providecommand{\BIBdecl}{\relax}
\BIBdecl

\bibitem{Han_Musing}
T.~S. Han, ``Musing upon information theory,'' XXXI Shannon Lecture, 2010,
  presented at IEEE Int. Sym. on Information Theory (ISIT).

\bibitem{Shannon_TheoryOfComm}
C.~E. Shannon, ``A mathematical theory of communication,'' \emph{Bell Systems
  Technical Journal}, vol.~27, pp. 379--423, 623--656, Oct. 1948.

\bibitem{Gray_Entropy}
R.~M. Gray, \emph{Entropy and Information Theory}.\hskip 1em plus 0.5em minus
  0.4em\relax New York, NY: Springer, 1990.

\bibitem{Pinsker_InfoEngl}
M.~S. Pinsker, \emph{Information and Information Stability of Random Variables
  and Processes}.\hskip 1em plus 0.5em minus 0.4em\relax San Francisco: Holden
  Day, 1964.

\bibitem{Watanabe_InfoLoss}
S.~Watanabe and C.~T. Abraham, ``Loss and recovery of information by coarse
  observation of stochastic chain,'' \emph{Information and Control}, vol.~3,
  no.~3, pp. 248--278, Sep. 1960.

\bibitem{Geiger_NLDyn1starXiv}
B.~C. Geiger and G.~Kubin, ``Some results on the information loss in dynamical
  systems,'' in \emph{Proc. IEEE Int. Sym. Wireless Communication Systems
  (ISWSC)}, Aachen, Nov. 2011, accepted; preprint available: {\tt
  arXiv:1106.2404 [cs.IT]}.

\bibitem{Geiger_ISIT2011arXiv}
B.~C. Geiger, C.~Feldbauer, and G.~Kubin, ``Information loss in static
  nonlinearities,'' in \emph{Proc. IEEE Int. Sym. Wireless Communication
  Systems (ISWSC)}, Aachen, Nov. 2011, accepted; preprint available: {\tt
  arXiv:1102.4794 [cs.IT]}.

\bibitem{Papoulis_Probability}
A.~Papoulis and U.~S. Pillai, \emph{Probability, Random Variables and
  Stochastic Processes}, 4th~ed.\hskip 1em plus 0.5em minus 0.4em\relax New
  York, NY: McGraw Hill, 2002.

\bibitem{Ruelle_EntropyProduction}
D.~Ruelle, ``Positiviy of entropy production in nonequilibrium statistical
  mechanics,'' \emph{J.Stat.Phys.}, vol.~85, pp. 1--23, 1996.

\bibitem{Rudin_Analysis3}
W.~Rudin, \emph{Real and Complex Analysis}, 3rd~ed.\hskip 1em plus 0.5em minus
  0.4em\relax New York, NY: McGraw-Hill, 1987.

\bibitem{Lazo_Entropies}
A.~C. Verdugo~Lazo and P.~N. Rathie, ``On the entropy of continuous probability
  distributions,'' \emph{IEEE Transactions on Information Theory}, vol. IT-24,
  pp. 120--122, 1978.

\bibitem{Abramowitz_Handbook}
M.~Abramowitz and I.~A. Stegun, Eds., \emph{Handbook of Mathematical Functions
  with Formulas, Graphs, and Mathematical Tables}, 9th~ed.\hskip 1em plus 0.5em
  minus 0.4em\relax Dover Publications, 1972.

\end{thebibliography}
